\documentclass[conference,10pt]{IEEEtran}
\IEEEoverridecommandlockouts
\usepackage{cite}
\usepackage{amsmath,amssymb,amsfonts}
\usepackage{amsthm}
\usepackage{algorithmic}
\usepackage{graphicx}
\graphicspath{{figures/}} 
\usepackage{textcomp}
\usepackage{xcolor}
\usepackage[a4paper, total={184mm,239mm}]{geometry}
\def\BibTeX{{\rm B\kern-.05em{\sc i\kern-.025em b}\kern-.08em
    T\kern-.1667em\lower.7ex\hbox{E}\kern-.125emX}}

\usepackage{siunitx}
\usepackage{todonotes}
\usepackage[capitalise]{cleveref}
\usepackage{comment}

\newtheorem{theorem}{Theorem}
\newtheorem{definition}[theorem]{Definition}
\newtheorem{lemma}[theorem]{Lemma}

\newcommand{\dup}{\delta_\uparrow}
\newcommand{\ddo}{\delta_\downarrow}
\newcommand{\dupdo}{\delta_{\uparrow/\downarrow}}

\newcommand{\dmin}{\delta_{\mathrm{min}}}

\newcommand{\dupinfty}{\delta^\uparrow_\infty}
\newcommand{\ddoinfty}{\delta^\downarrow_\infty}

\newcommand{\idm}{\text{IDM}}
\newcommand{\cidm}{\text{CIDM}}
\newcommand{\etam}{$\eta$-\text{IDM}}
\newcommand{\etacidm}{$\eta$-\text{CIDM}}

\newcommand{\spice}{\text{SPICE}}

\newcommand{\vth}{V_{th}}
\newcommand{\vthin}{V_{th}^{in}}

\newcommand{\vdd}{V_{DD}}

\newcommand{\etp}{\eta^+}

\newcommand{\etpd}{\eta^{+'}}
\newcommand{\etm}{\eta^-}
\newcommand{\etmd}{\eta^{-'}}

\newcommand{\etpinfty}{\eta^+_\infty}
\newcommand{\etminfty}{\eta^-_\infty}

\newcommand{\etpmin}{\eta^+_{min}}
\newcommand{\etminf}{\eta^-_{\infty}}
\newcommand{\etpinf}{\eta^+_{\infty}}

\newcommand{\etmmin}{\eta^-_{min}}

\newcommand{\dupsx}[1]{{\delta_\uparrow^*}_{#1}}
\newcommand{\ddosx}[1]{{\delta_\downarrow^*}_{#1}}
\newcommand{\dupsdef}{\dupsx{def}}
\newcommand{\ddosdef}{\ddosx{def}}

\newcommand{\cDelta}{\Delta'}
\newcommand{\cP}{P'}

\newcommand{\rhp}{\rho^+}
\newcommand{\rhm}{\rho^-}

\newcommand{\bdup}{\overline{\delta}_\uparrow}
\newcommand{\bddo}{\overline{\delta}_\downarrow}

\newcommand{\vthinm}{V_{th}^{in*}}
\newcommand{\vthoutm}{V_{th}^{out*}}

\newcommand{\oDelta}{\overline{\Delta}}
\newcommand{\Deltapm}{\Delta^{+/-}}

\newboolean{conference}
\setboolean{conference}{false}

\newcommand\copyrighttext{%
	\footnotesize \textcopyright 2023 IEEE. Personal use of this material is 
	permitted.
	Permission from IEEE must be obtained for all other uses, in any current or 
	future 
	media, including reprinting/republishing this material for advertising or 
	promotional 
	purposes, creating new collective works, for resale or redistribution to 
	servers or 
	lists, or reuse of any copyrighted component of this work in other works. 
}

\newcommand\copyrightnotice{%
	\begin{tikzpicture}[remember picture,overlay]
		\node[anchor=south,yshift=10pt] at (current page.south) 
		{\fbox{\parbox{\dimexpr\textwidth-\fboxsep-\fboxrule\relax}{\copyrighttext}}};
	\end{tikzpicture}%
}

\begin{document}
%

\title{A Digital Delay Model Supporting Large Adversarial Delay 
Variations
\thanks{This research was supported by the Austrian Science Fund (FWF) project 
DMAC (grant no. P32431).}
}

\author{
	\IEEEauthorblockN{
		Daniel \"Ohlinger, 
		Ulrich Schmid
	}
	\IEEEauthorblockA{Embedded Computing Systems Group (E191-02)\\
		TU Wien, Vienna, Austria \\
		\{doehlinger, s\}{@}ecs.tuwien.ac.at}
}

\maketitle

\copyrightnotice

\begin{abstract}
Dynamic digital timing analysis is a promising alternative to analog simulations
for verifying particularly timing-critical parts of a circuit. A necessary
prerequisite is a digital delay model, which allows to accurately predict the
input-to-output delay of a given transition in the input signal(s) of a gate.
Since all existing digital delay models for dynamic digital timing analysis
are deterministic, however, they cannot cover delay fluctuations caused by PVT variations,
aging and analog signal noise. The only exception known to us is the \etam\ 
introduced
by F\"ugger et~al.~ at DATE'18, which allows to add (very) small adversarially
chosen delay variations to the deterministic involution delay model, without
endangering its faithfulness. In this paper, we show that it is possible to 
extend the range of allowed delay variations so significantly that 
realistic PVT variations and aging are covered by the resulting extended \etam.
\end{abstract}


\section{Introduction}\label{sec:intro}

Accurate signal propagation predictions are crucial for
modern digital circuit design. The highest accuracy is currently achievable by
analog simulations, e.g., using \spice. These suffer, however, from excessive
running times. A considerably more efficient alternative is
\emph{dynamic} digital timing
analysis, which traces individual signal transitions throughout a
circuit. Application examples are clock trees
or time-based encoded inter-neuron links in hardware-implemented spiking neural networks \cite{BVMRRVB19},
where the (very accurate but worst-case) delay estimates provided by
classic static timing 
analysis techniques like CCSM~\cite{Syn16:CCSM} and ECSM~\cite{Cad15:ECSM} are not sufficient 
for ensuring correct operation.

Existing digital dynamic timing analysis tools like Cadence NC-Sim,
Siemens ModelSim or Synopsis VCS rely on simple gate delay
models like pure delays (= constant delay) or 
inertial delays (= constant delay, but pulses shorter 
than some upper bound are removed)~\cite{Ung71}. However,
research has also provided more elaborate \emph{single-history} 
delay models like the \emph{degradation delay model} (DDM) \cite{BJV06} 
and, in particular, the \emph{involution delay model} (IDM) \cite{FNNS15:DATE,FNNS20:TCAD}, where 
the input-to-output delay 
$\delta(T)$ depends on the previous-output-to-input delay $T$ (see \cref{fig:single-history}).

The distinguishing feature of the IDM is that it is the only delay model
known so far that is faithful w.r.t.\ a ``canonical'' problem called 
\emph{short-pulse filtration} (SPF) 
\cite{FNS16:TCOM}\ifthenelse{\boolean{conference}}{,}{ (see \cref{def:SPF}),}
which means that 
the solvability/impossibility border for circuits specified in the IDM
matches the solvability/impossibility border for real circuits:
whereas it is possible to implement unbounded SPF in
reality, this is not the case for bounded SPF, where all output
transitions must occur within bounded time.
An essential property required for faithfulness is \emph{continuity} of the 
underlying digital channel model, in the sense that minor changes (which also include 
short glitches, at arbitrary positions) of the signal at an input of a gate 
must result in minor changes of the output signal only. It has been proved in \cite{FNS16:TCOM}
that no existing delay model (except the IDM) satisfies this continuity
property.

The defining characteristics of the IDM is a pair of differentiable concave
delay functions (monotonically increasing, with monotonically decreasing
derivative) for rising ($\dup(T)$) and falling ($\ddo(T)$) transitions
that satisfy the involution property $-\dup(-\ddo(T))=T$.
It has been proved in \cite{FNNS20:TCAD} that
this ensures continuity and hence faithfulness w.r.t.\ the SPF problem.
The IDM also comes with public domain tool support, the \emph{Involution Tool}, a simulation framework 
based on ModelSim, which has been used to demonstrate the good accuracy of 
the IDM predictions for several simple circuits~\cite{OMFS21:Integration}.
Certain deficiencies of the IDM regarding composition of circuits have been
alleviated in the \emph{Composable Involution Delay Model} (CIDM)
\cite{MOSFN21:GLSVLSI}.

One shortcoming of all existing delay models, including the original IDM and
the CIDM,
are their \emph{deterministic} delay predictions. After all, e.g.\ $\dup(T)$ is a
\emph{function} that only depends on the previous-output-to-input delay $T$. 
Delay fluctuations caused by other effects, like
PVT variations, aging, and electrical signal noise, cannot be accommodated here.
Given the prominence of these effects in modern VLSI circuits, however, 
they should be incorporated somehow. We note that this requirement is even 
more pressing in any attempt on formal verification of circuits, which aims at proving 
that a given circuit will meet its specification in \emph{every} feasible trace.

In their DATE'18 paper \cite{FMNNS18:DATE}, F\"ugger at.~al.\ introduced an 
extension of the IDM, called \etam,
which allows to add a small amount of adversarial delay variations to the 
fully deterministic delay predictions of the \idm. That is, for every signal transition,
the standard IDM delay $\dupdo(T)$ can be replaced by $\dupdo(T)+\eta$ for some
arbitrarily chosen $\eta \in [-\etm,\etp]$. This feature can of course
be used to cover any bounded-range delay fluctuation. The authors proved 
that the resulting \etam\ does not invalidate the faithfulness of the original IDM, 
provided the variation range $[-\etm,\etp]$ is (very) small. Unfortunately, however,
\spice\ simulation data revealed that even voltage variations of 1~\% and transistor
size variations of 10~\% already exceeded the allowed range.

In this paper, we show that the allowed range $[-\etm,\etp]$ can be extended
significantly. The key to our \emph{extended \etam} is to make the allowed variation range
dependent on $T$, i.e., to allow $\eta \in [-\etm(T),\etp(T)]$ depending on
the previous-output-to-input delay $T$ of the current transition. In particular,
except for very small values of $T$, the variation range may be large. 

\paragraph*{Detailed contributions}
\begin{enumerate}
\item[(1)] We define our extended \etam, by providing the constraints that
  must be satisfied by $[-\etm(T),\etp(T)]$, and prove that the
  SPF implementation already used in \cite{FMNNS18:DATE} also works
  correctly in this model.
\item[(2)] We show that it is possible to combine the extended \etam\ with the
\cidm\ \cite{MOSFN21:GLSVLSI}, which results in the \etacidm\ that provides the best of both worlds.
We also extended the Involution Tool \cite{OMFS21:Integration} appropriately.
\item[(3)] We perform elaborate analog simulations under different PVT 
variations and aging, and compare the simulated delays with the 
predicted delays of our \etacidm. 
It turns out that, unlike the model in 
\cite{FMNNS18:DATE}, our new model can capture a wide
range of such variations.
\end{enumerate} 

\paragraph*{Paper organization} In \cref{sec:etaidm}, the necessary basics for 
the \idm\ and the \etam\ are presented.
\cref{sec:extension} provides the proofs for the faithfulness for the extended
delay variation range. In \cref{sec:etacidm}, we sketch how the extended \etam\ and
the \cidm\ can be seamlessly combined. The experiments in
\cref{sec:results} demonstrate
the substantially increased coverage of our new model.
\cref{sec:conclusion} concludes our paper.

\section{The Existing \etam}\label{sec:etaidm}

In sharp contrast to all other delay models known so far, the
\idm\ is based on \emph{unbounded} single-history channels: Referring to
\cref{fig:single-history}, the gate delay $\delta(T)$ may also
become arbitrarily negative here (for specific values of the
previous-output-to-input delay $T$), which has been proved to be
mandatory for perfect glitch cancellation \cite{FNS16:TCOM}.

\begin{figure}[tbp]
	\centering
	\ifthenelse{\boolean{conference}}
	{
		\includegraphics[width=.70\linewidth]{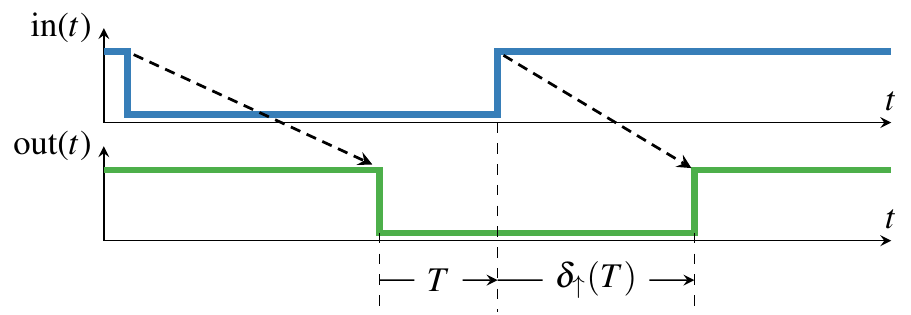}
	}
	{
		\includegraphics[width=.85\linewidth]{single_history.pdf}
	}
	\caption{Delay calculation for a single-history channel.}
	\label{fig:single-history}
\end{figure}

In order to also account for PVT variations and aging, \cite{FMNNS18:DATE} introduced
the \etam:
On top of the deterministic delay calculations of the \idm,
an adversarial delay $\eta \in 
[-\etmmin, \etpmin]$ is added. More specifically, 
the delay $\delta_n$ for the $n^{th}$ transition is calculated as 
\begin{equation}
\delta_n = \dup(\max\{{t_n - t_{n-1} - \delta_{n-1}, -\ddoinfty}\}) + \eta_n
\end{equation}
 for a rising transition, and
\begin{equation}
\delta_n = \ddo(\max\{{t_n - t_{n-1} - \delta_{n-1}, -\dupinfty}\}) + 	\eta_n
\end{equation}
for a falling transition, where $t_n$ and $t_{n-1}$ are the time of the 
current resp. previous input transition, and $\eta_n \in 
[-\etmmin, \etpmin]$.
\cref{fig:etaidmtrace} illustrates two of the infinitely many possible
output signals of a channel in the \etam, for a fixed input signal, 
where the dashed lines represent the deterministic 
delay prediction by the \idm.
It is apparent that the adversary has the power to cancel or uncancel 
pulses w.r.t.\ the deterministic prediction. 

\begin{figure}[h]
	\centering
	\ifthenelse{\boolean{conference}}
	{
		\includegraphics[width=0.8\linewidth]{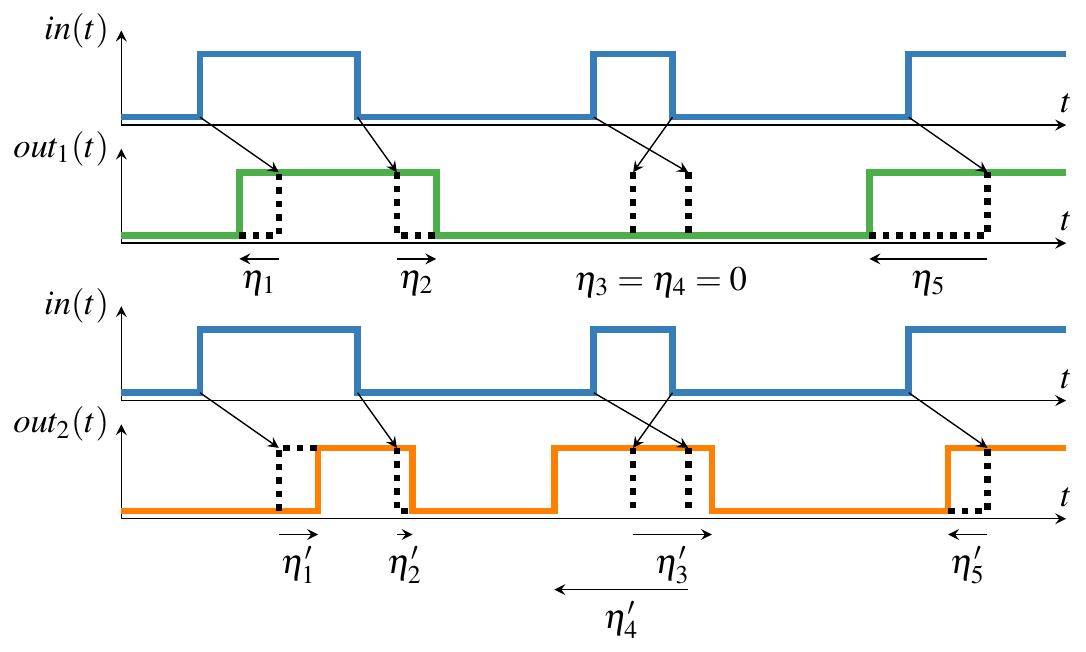}
	}
	{
		\includegraphics[width=0.9\linewidth]{eta_example_trace.pdf}
	}
	\caption{Example trace through an $\eta$-involution channel (adapted from 
		\cite{FMNNS18:DATE}).}
	\label{fig:etaidmtrace} 
\end{figure}

\ifthenelse{\boolean{conference}}
{Informally, the \emph{Short-Pulse Filtration} (SPF) problem is similar to 
building an inertial delay channel from a single input up-down pulse of 
length $\Delta_0$ starting at time $t_0$, see \cite[Def.~1]{FMNNS18:DATE} for its
formal definition.}
{
Since faithfulness is defined with respect to the \emph{Short-Pulse Filtration} (SPF) problem, we recall its definition here. Informally, it is 
similar to the problem of building an inertial delay channel from a single
input up-down pulse of length $\Delta_0$ starting at time $t_0$:

\begin{definition}[Short-Pulse Filtration, see \cite{FNNS20:TCAD}]\label{def:SPF}	
	A circuit that solves SPF needs to adhere to the following constraints:
	\begin{enumerate}
		\item[F1)] Well-formedness: The circuit has exactly one input and 
		one output port.
		\item[F2)] No generation: The circuit does not generate a pulse on 
		the output, if no pulse is on the input.
		\item[F3)] Nontriviality: The output is not the zero signal for all 
		pulses.
		\item[F4)] No short pulses: There exists an $\varepsilon > 0$ such 
		that the output signal does not contain pulses of length less than 
		$\varepsilon$.
	\end{enumerate}	
	For the bounded version of SPF, there is an additional constraint:	
	\begin{enumerate}
		\item[F5)] Bounded stabilization time: There exists a $K > 0$ such that 
		for every input pulse the last output transition is before time $t_0 + K$.
	\end{enumerate}
\end{definition}
For any input signal other than a single pulse, we allow the SPF circuit
to behave arbitrarily.

}
In order to show that the \etam\ does not allow to solve bounded SPF,
the authors of \cite{FMNNS18:DATE} used a simple reduction proof:
By setting all delay variations $\eta = 0$, the 
model degenerates to the plain \idm, where this is known to be
impossible \cite{FNNS20:TCAD}.
The other case, namely, showing that it is possible to
solve unbounded SPF in the \etam, turned out to be trickier:
For the SPF circuit given in \cref{fig:spfcircuiteta}, they analyzed 
all possible behaviors via a case distinction on the length of the input pulse
$\Delta_0$ (short, large and medium).
The most delicate case turned out to be medium pulse lengths,
where a critical (infinite)
pulse train with up-time $\Delta$ and down-time $\Delta'$ may be generated.
It occurs if the adversary maximally delays all rising transitions
(by $\etpmin$) and minimally delays all falling transitions (by $-\etmmin$),
provided the (quite restrictive) constraint 
\begin{equation}
	\etpmin + \etmmin < \ddo(-\etpmin) - \dmin\tag{C1}\label{eq:constraintc1}
\end{equation}
holds. If the adversary chooses delay variations that ever cause the
up-time of a pulse to become larger than $\Delta$, it cannot
revert to pulses with shorter up-times again, and hence cannot prohibit
the feedback loop of the OR gate to eventually lock at 1. Consequently,
a correct SPF circuit is obtained by letting the high-threshold
buffer at the output map all pulse trains with up-times less or
equal $\Delta$ to constant $0$, as all pulse trains that ever
exceed $\Delta$ will lead to a single rising transition at the output.

\begin{figure}[t]
	\centering	
	\ifthenelse{\boolean{conference}}
	{
		\includegraphics[width=0.8\linewidth]{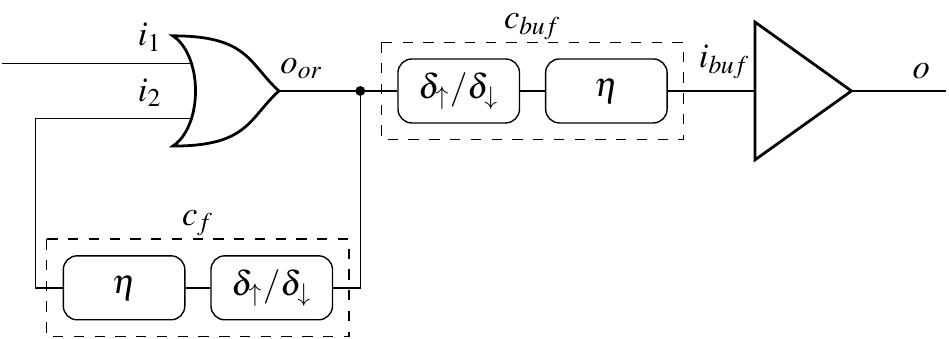}
	}
	{
		\includegraphics[width=0.9\linewidth]{unbounded_spf_circuit_etaidm.pdf}
	}
	\caption{The SPF circuit for the \etam\ used in \cite{FMNNS18:DATE}. It  consists of an OR gate with a feed-back \etam\ channel ($c_f$) and
          a subsequent high-threshold buffer implemented by another
          \etam\ channel ($c_{buf}$).}
	\label{fig:spfcircuiteta}
\end{figure}

\section{Extending the Variation Range of the \etam}\label{sec:extension}

In this section, we will show that the range of adversarial variations of the
\etam\ can be increased substantially, without sacrificing 
faithfulness with respect to the SPF problem. The crucial idea of
our extension is to introduce $T$-dependent bounds in the variation range
$[-\etm(T), \etp(T)]$,
which coincide with the original bounds $[-\etmmin,\etpmin]$ only
for two very small values of $T$, which are quite insensitive 
to delay fluctuations.

Let $\Delta$ be the up-time of the critical pulse train identified in \cite{FMNNS18:DATE} (which is determined by the fixed point of $f(.)$ given
in \cref{eq:delta-new}) and $\dmin>0$ be the unique minimum
delay value of a strictly causal\footnote{Where $\dup(0)>0$ and $\ddo(0)>0$.} 
\idm\ channel,
defined by
\begin{equation}
  \dup(-\dmin) = \dmin = \ddo(-\dmin)\label{eq:dmin}
\end{equation}
according to~\cite[Lem.~3]{FNNS20:TCAD}. 

\begin{definition}\label{def:eta}
For some given constants
	\begin{equation}
		\rhp \geq 0, \rhm \geq 0\tag{C2}\label{eq:constraintc2}
	\end{equation}
some $\etpinf\geq \etpmin$, $\etminf \geq \etmmin$,
some $\oDelta$ satisfying $\oDelta \geq \dup(-\oDelta) + \rho^+ \cdot (\oDelta - \Delta) + 
	\etpmin$, 
and 	$\Delta' = \dup(-\Delta) + \etpmin - \Delta$, let	
	
	\begin{equation}
		\etp(T) =
		\begin{cases}
			\etpinf  & 
			\text{for}\ T < -\overline{\Delta},\\
			\rhp \cdot (-T - \Delta) + \etpmin & 
			\text{for}\ -\oDelta \leq T \leq -\Delta,\\
			\etpinf & 
			\text{for}\ T > -\Delta,\label{eq:etapdef}
		\end{cases}
	\end{equation}

	\begin{equation}
		\etm(T) =
		\begin{cases}
			\etminf  & 
			\text{for}\ T < -\Delta',\\
			\rhm \cdot (-T + \Delta') + \etmmin & 
			\text{for}\ - \Delta' \leq T  
			< 0,\\
			\etminf & 
			\text{for}\ T \geq 0.\label{eq:etamdef}
		\end{cases}
	\end{equation}
\end{definition}

\cref{fig:boundcomp} shows the comparison of the old and the new bounds for the 
\etam. Note carefully that the only values of $T$
for which the bounds could not be enlarged are $-\Delta$ and
$-\Delta'$, i.e., the up-time and the
down-time of the critical infinite pulse train
identified in \cite{FMNNS18:DATE}.
For the inverter used in \cref{sec:results}, the actual values are 
		$\etpmin \approx \etmmin \approx \SI{63}{\fs}, 
		\etpinf \approx \etminf \approx \SI{1.47}{\ns}, 
		\Delta \approx \SI{191}{\fs}, 
		\Delta' \approx \SI{63}{\fs}$ and 
		$\oDelta = \SI{230}{\fs}$.
\ifthenelse{\boolean{conference}}{}{Both $\etpinf$ and $\etminf$ are hence approximately $20$ 
times larger than the original bounds $\etpmin$ and $\etmmin$ here.}

\begin{figure}[t]
	\centering
	\ifthenelse{\boolean{conference}}
	{
		\includegraphics[width=0.90\linewidth]{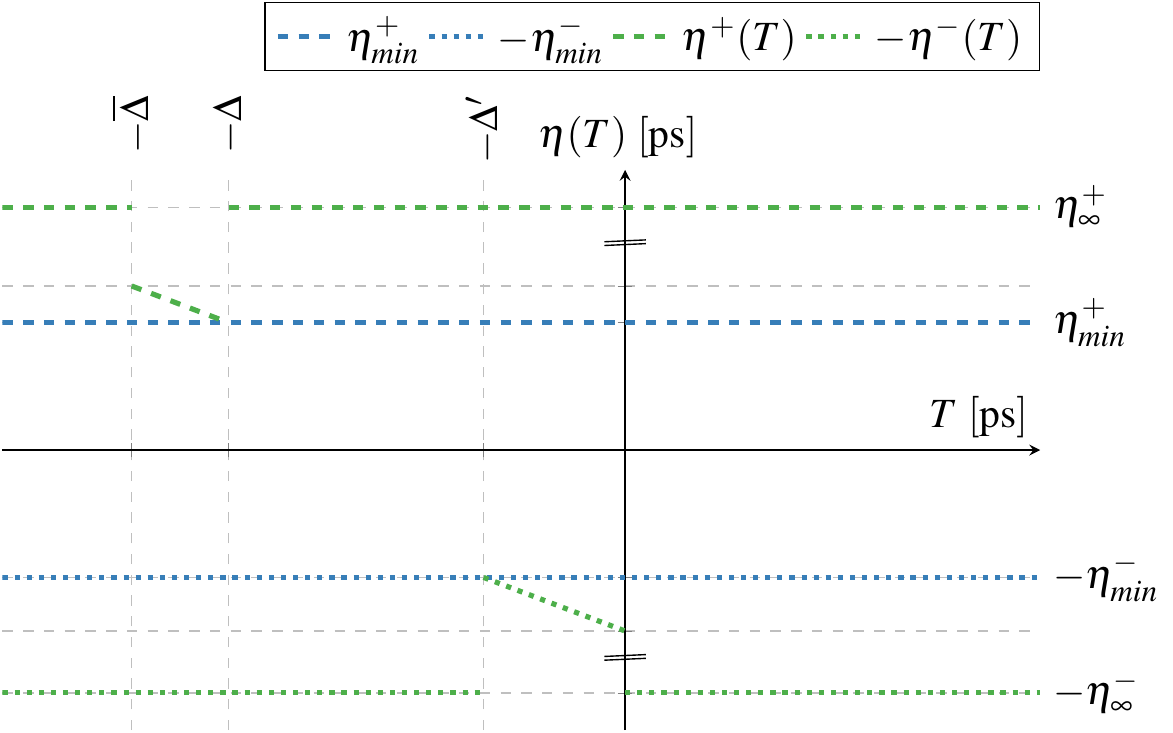}
	}
	{
		\includegraphics[width=1\linewidth]{boundary_idea_new.pdf}
	}
	\caption{Comparison of the old bounds (blue) and the new 
	  bounds (green). Note the discontinuity sign
          on the y-axis. 
		For 
		the example circuit used in our evaluation in \cref{sec:results}, $\etpinfty \approx \etminf 
		\approx 20 \etpmin \approx 20 \etmmin$.		
	}
	\label{fig:boundcomp} 
\end{figure}

For proving the faithfulness of the extended \etam,
we show:
\begin{itemize}
	\item[G1)] There is no circuit with extended \etam-channels that can solve bounded 
	SPF.
	\item[G2)] There exists a circuit with extended \etam-channels that can solve 
	unbounded 
	SPF.
\end{itemize}

For G1), the exact same proof as in the original \etam\ paper 
\cite{FMNNS18:DATE}
can be used: Setting all the adversarial delays $\eta_n = 0$ lets the 
\etam-channel degenerate to an \idm-channel, which does not allow
solve bounded SPF \cite{FNNS20:TCAD}.

The intricate part is again the proof of G2), where we
split the range of the input pulse lengths $\Delta_0$ into
three intervals. 

\begin{lemma}\label{lemma:longpulses}
  If the input pulse length satisfies $\Delta_0 \geq \dupinfty + \etpinfty$ for
  $\dupinfty = \lim_{T\to\infty}\dup(T)$, 
	then the output of the OR gate $o_{or}$ in \cref{fig:spfcircuiteta} has a 
	rising transition at time $0$ and no falling transitions.
\end{lemma}

\ifthenelse{\boolean{conference}}
{
	\begin{proof}
		The rising transition from the input $i_1$ arrives at latest at $t_1' 
		\leq 
		\dupinfty + \etpinfty$ at $i_2$ (one can assume that the OR gate has
		been initialized to output 0 at time $t=-\infty$).
		Therefore, the falling transition on 
		input $i_1$ happens at or after $t_1'$, so has no influence on the 
		output of the OR gate: the feedback loop hence locks at constant $1$.
	\end{proof}
}
{
	\begin{proof}
		The rising transition from the input $i_1$ arrives at latest at $t_1' 
		\leq 
		\dupinfty + \etpinfty$ at $i_2$ (one can assume that the OR gate has
		been initialized to output 0 at time $t=-\infty$).
		Therefore, the falling transition on 
		input $i_1$ happens at or after $t_1'$, so has no influence on the 
		output of the OR gate: the feedback loop hence locks at constant $1$.
	\end{proof}
}

\begin{lemma}\label{lemma:shortpulses}
	If the input pulse length satisfies $\Delta_0 \leq \dupinfty - \dmin - 
	\etpinfty - \etminfty$, then the output pulse of the OR gate $o_{or}$ in 
	\cref{fig:spfcircuiteta} only contains the input pulse.
\end{lemma}

\ifthenelse{\boolean{conference}}
{
     \begin{proof}
     Lacking space did not allow us to include all the proofs
     in this paper. They are provided in an extended version 
     \cite{OS23:arxiv}.
     \end{proof}
}
{	
	\begin{proof}
		The earliest time that the rising transition from $i_1$ can be 
		propagated 
		to $i_2$ is now $t_1' \geq \dupinfty - \etminfty$. 
		Therefore, for the falling transition, we get $T = \Delta_0 - t_1'$,  
		the falling output transition on $i_2$ hence cannot occur later than 
		$t_2' 
		\leq \Delta_0 + \ddo(T) + \etp(T)$. 
		The two transitions on $i_2$ cancel each other out iff $t_2' \leq 
		t_1'$, 
		i.e., iff
		\begin{equation}\label{eq:shortpulse}
			X = \Delta_0 + \ddo(T) + \etp(T) - \dupinfty + \etminfty \leq 0.
		\end{equation}
		
		Plugging in the upper bound on $\Delta_0$ from our lemma into
		$\ddo(T)=\ddo(\Delta_0-t_1') \leq \ddo(\Delta_0-\dupinfty + \etminfty)$
		while recalling monotonicity of $\ddo(.)$, $\etpinfty\geq 0$ reveals
		$\ddo(T) \leq -\dmin + \ddo(-\dmin) = 0$ according to \cref{eq:dmin}. 
		Using these upper bounds
		in \cref{eq:shortpulse} leads to $X\leq 0$, hence the two transitions
		indeed cancel each other.
	\end{proof}
}

Note that the upper bound on $\Delta_0>0$ given in \cref{lemma:shortpulses}
already implies the constraint 
\begin{equation}
 	\etpinf + \etminf < \dupinfty - \dmin.\tag{C3}\label{eq:constraintc3}
\end{equation}

The most delicate part of the proof are again medium-sized 
input pulses, i.e., where $\dupinfty - \dmin - \etpinf - \etminf < \Delta_0 < 
\dupinfty + \etpinf$. Inspired by \cite{FMNNS18:DATE}, we also consider
a self-repeating critical pulse train, by 
letting the adversary choose all rising transitions maximally late 
($\etp(T)$) and all falling transitions maximally early ($-\etm(T)$).
Formally, for a pulse of length $\Delta_{n-1}$, the length of the next pulse $\Delta_{n}$
can be calculated as:
\begin{align}
	\Delta_n = &g(\Delta_{n-1}) = \nonumber\\
	&\Delta_{n-1} - \dup(-\Delta_{n-1}) 
	- \etp(-\Delta_{n-1})\nonumber\\
	&+ \ddo(-\dup(-\Delta_{n-1}) - \etp(-\Delta_{n-1}) + 	
	\Delta_{n-1})\nonumber\\ 
	&- \etm(-\dup(-\Delta_{n-1}) - \etp(-\Delta_{n-1}) + \Delta_{n-1}), 
	\label{eq:delta-new}
\end{align}
where $\etp(T)$ and $\etm(T)$ are given by \cref{def:eta}.
The period $P_n=\Delta_n+\cDelta_{n+1}$ resp.\ $\cP_{n}=\cDelta_n+\Delta_n$
involving the pulse $\Delta_n$, measured
from the rising transition of $\Delta_n$ to the rising transition
of $\Delta_{n+1}$ resp.\ the falling transition of $\Delta_{n-1}$ to
the falling transition of $\Delta_{n}$, is
\begin{align}
  P_n &= \dup(-\Delta_n)+\etp(-\Delta_{n}) \mbox{ resp.\ }\nonumber\\
  \cP_n &= \ddo(-\cDelta_n)-\etm(-\cDelta_{n})\label{def:periods-new}.
\end{align}
We note that the length of the next pulse and the periods
of the critical pulse train established in \cite{FMNNS18:DATE} are
\begin{align}
	\Delta_n &= f(\Delta_{n-1}) = \Delta_{n-1} 
	- \dup(-\Delta_{n-1}) - \etpmin - \etmmin\nonumber\\ 
	&\qquad\qquad + \ddo\Big(-\dup(-\Delta_{n-1}) - \etpmin + 
	\Delta_{n-1}\Big)\label{eq:delta-old},\\
	P_n &=   \dup(-\Delta_n)+\etpmin\nonumber \mbox{ resp.\ }
\cP_n = \ddo(-\cDelta_n)-\etmmin\nonumber,
\end{align}
which were shown in \cite[Lem.~5]{FMNNS:arxiv} to have a fixed point
$\Delta:=\Delta_{n-1}=\Delta_n < \dmin$ and a period $P=\cP=\tau$ for
some $\etpmin+\dmin < \tau < \min(-\etmmin + \delta^\downarrow_\infty, \etpmin + \delta^\uparrow_\infty)$, provided the constraint
\cref{eq:constraintc1} holds.

The following \cref{lemma:fixedpoint}
shows that $\Delta$
is also a fixed point for \cref{eq:delta-new}:

\begin{lemma}\label{lemma:fixedpoint}
  The fixed point $f(\Delta) = \Delta$ for the old bounds is
  also a fixed point for the new bounds, i.e., $g(\Delta) = \Delta$.
  The critical pulse train defined by \cref{eq:delta-new} hence
  has the same period $P=P'=\tau$ and duty cycle $\gamma=\Delta/\tau < 1$
  as the one defined by \cref{eq:delta-old}.
\end{lemma}

\ifthenelse{\boolean{conference}}
{
	\begin{proof}
		Plugging in $\Delta$ in \cref{eq:delta-new} and applying the definition 
		for 
		$\etp(.)$ and $\etm(.)$ from \cref{def:eta} yields
		\begin{align}
			g(\Delta) = &\ddo\Big(-\dup(-\Delta) - \etpmin + 
			\Delta\Big)\nonumber\\ 
			&- \etmmin - \dup(-\Delta) - \etpmin + \Delta.\label{eq:fixpointg}
		\end{align}
		
		A term-wise comparison of \cref{eq:delta-old} with 
		$\Delta_{n-1}=\Delta_n=\Delta$ and \cref{eq:fixpointg} 
		reveals that $g(\Delta) = \Delta$. Plugging in $\Delta$ in
		\cref{def:periods-new} and noting $\etp(-\Delta)=\etpmin$ and
		$\etm(-\Delta)=\etmmin$ also confirms the same period and hence
		duty cycle.
	\end{proof}
}
{
	\begin{proof}
		Plugging in $\Delta$ in \cref{eq:delta-new} and applying the definition 
		for 
		$\etp(.)$ and $\etm(.)$ from \cref{def:eta} yields
		\begin{align}
			g(\Delta) = &\ddo\Big(-\dup(-\Delta) - \etpmin + 
			\Delta\Big)\nonumber\\ 
			&- \etmmin - \dup(-\Delta) - \etpmin + \Delta.\label{eq:fixpointg}
		\end{align}
		
		A term-wise comparison of \cref{eq:delta-old} with 
		$\Delta_{n-1}=\Delta_n=\Delta$ and \cref{eq:fixpointg} 
		reveals that $g(\Delta) = \Delta$. Plugging in $\Delta$ in
		\cref{def:periods-new} and noting $\etp(-\Delta)=\etpmin$ and
		$\etm(-\Delta)=\etmmin$ also confirms the same period and hence
		duty cycle.
	\end{proof}
}

In order to guarantee diverging from the critical pulse
train for excessive pulses, we need a bound on the derivative of $g(.)$:

\begin{lemma}\label{lemma:newboundsderivative} For $g(x)$ defined in
  \cref{eq:delta-new}, we have $g'(x) > 1$ for $\oDelta 
	\geq x \geq \Delta$, provided
	\begin{equation}
		(1 - \rhm)(\dup'(-\Delta) - \rhp + 1) > 		
		1.\tag{C4}\label{eq:constraintc4}
	\end{equation}
\end{lemma}

\ifthenelse{\boolean{conference}}
{
	\begin{proof}
		Calculating the derivative of $g(x)$ yields:
		\begin{align}
			g'(x) = 
			& \bigg[1 + \ddo'\Big(-\dup(-x) - \etp(-x) + x\Big)\nonumber\\
			&- \etmd\Big(-\dup(-x) - \etp(-x) + 
			x\Big)\bigg]\nonumber\\
			&\cdot \Big(\dup'(-x) + \etpd(-x) + 1\Big)\nonumber\\
			\geq& \bigg[1 + \ddo'\Big(-\dup(-x) - \etp(-x) + 
			x\Big)- \rho^-\bigg]\nonumber\\
			&\cdot \Big(\dup'(-x) - \rho^+ + 1\Big)\label{eq:gderivativebound}
			\\>&(1 - \rho^-) (\dup'(-\Delta) - \rho^+ + 1) > 1
		\end{align}
		To justify \cref{eq:gderivativebound}, we note that $\dup(.)$ and 
		$\ddo(.)$ are strictly increasing and concave, so $\dup'(.) > 0$ and 
		$\ddo'(.) > 0$. Moreover, $\dup'(-x)$ is decreasing in $-x$, hence can 
		be lower
		bounded by $\dup'(-\Delta)$ since $x \geq \Delta$, and
		$\rhp \geq 0$ and $\rhm \geq 0$ according to \cref{def:eta}.	
		Choosing $\rhp$ and $\rhm$ according to the sufficient 
		condition \cref{eq:constraintc4} finally ensures $g'(x) > 1$.
	\end{proof}
}
{
	\begin{proof}
		Calculating the derivative of $g(x)$ yields:
		\begin{align}
			g'(x) = 
			& \bigg[1 + \ddo'\Big(-\dup(-x) - \etp(-x) + x\Big)\nonumber\\
			&- \etmd\Big(-\dup(-x) - \etp(-x) + 
			x\Big)\bigg]\nonumber\\
			&\cdot \Big(\dup'(-x) + \etpd(-x) + 1\Big)\nonumber\\
			\geq& \bigg[1 + \ddo'\Big(-\dup(-x) - \etp(-x) + 
			x\Big)- \rho^-\bigg]\nonumber\\
			&\cdot \Big(\dup'(-x) - \rho^+ + 1\Big)\label{eq:gderivativebound}
			\\>&(1 - \rho^-) (\dup'(-\Delta) - \rho^+ + 1) > 1
		\end{align}
		To justify \cref{eq:gderivativebound}, we note that $\dup(.)$ and 
		$\ddo(.)$ are strictly increasing and concave, so $\dup'(.) > 0$ and 
		$\ddo'(.) > 0$. Moreover, $\dup'(-x)$ is decreasing in $-x$, hence can 
		be lower
		bounded by $\dup'(-\Delta)$ since $x \geq \Delta$, and
		$\rhp \geq 0$ and $\rhm \geq 0$ according to \cref{def:eta}.	
		Choosing $\rhp$ and $\rhm$ according to the sufficient 
		condition \cref{eq:constraintc4} finally ensures $g'(x) > 1$.
	\end{proof}
}

As the next step in our proof, we need to show that pulses that are larger than 
$\Delta$ cannot come back, i.e., if $\Delta_{n-1}>\Delta$ ever holds, then
$\Delta_{n-1+k} < \Delta_{n+k}$ for every $k\geq 0$:

\begin{lemma}\label{lemma:monincreasing}
	It holds that $g(\Delta_1) - \Delta > \Delta_1 - \Delta$ if $\Delta_1 > 
	\Delta$, provided that $\rhp$ and $\rhm$ are chosen according to 
	\cref{eq:constraintc4}. 
\end{lemma}

\ifthenelse{\boolean{conference}}
{
}
{
	\begin{proof}	
		From the mean value theorem of calculus, we know that 
		\begin{align}
			\exists \xi \in (\Delta, \Delta_1)\ s.t.\ g'(\xi) = 
			\frac{g(\Delta_1) 
				- g(\Delta)}{\Delta_1-\Delta}.
		\end{align}
		Applying the fact that $g(\Delta) = \Delta$ and rearranging yields 
		$g(\Delta_1) - \Delta = g'(\xi)(\Delta_1 - \Delta)$.
		From \cref{lemma:newboundsderivative}, we know that 
		$g'(\xi) > (1 - \rhm)(\dup'(-\Delta) - \rhp + 1)$, and hence 
		$g(\Delta_1) - 
		\Delta > \Delta_1 - \Delta$, if $\rhp$ and $\rhm$ are chosen 
		appropriately.
	\end{proof}
}

It still remains to be shown that the feedback loop locks at 1 when the pulse width 
reaches $\oDelta$ given in \cref{def:eta}:

\begin{lemma}\label{lemma:feedbacklock}
  The feedback loop locks when the pulse length satisfies $\Delta_n \geq
  \oDelta \geq 
	\dup(-\oDelta) + \rhp \cdot (\oDelta - \Delta) + \etpmin$.
\end{lemma}

\ifthenelse{\boolean{conference}}
{
}
{
	\begin{proof}
		The feedback loops locks at 1 if the rising transition of pulse
		$\Delta_{n+1}$ is scheduled before or at 
		the same time as the falling transition of pulse $\Delta_n$. 
		This cancellation happens when $\dup(-\oDelta) + \etp(-\oDelta) \leq 
		\oDelta$, which is implied by the assumption on $\oDelta$ stated in our 
		lemma
		and the definition of $\etp(T)$ for $T=-\oDelta$ given in 
		\cref{def:eta}.
	\end{proof}
}

\begin{theorem}\label{thm:mediumpulses}
	The new bounds ensure that pulses which are larger than $\Delta$ are 
	monotonically increasing until the pulse length reaches $\oDelta$ and the 
	feedback loop locks at 1.
\end{theorem}

\begin{proof}
	The theorem immediately follows by combining 	
	\cref{lemma:fixedpoint,lemma:newboundsderivative,lemma:monincreasing,lemma:feedbacklock}.
\end{proof}	

\begin{theorem}\label{thm:casedistinction}
	Consider the circuit in \cref{fig:spfcircuiteta} subject to the following 
	constraints:
	\begin{enumerate}
		\item[(C1)] $\etpmin + \etmmin < \ddo(-\etpmin) - \dmin$
		\item[(C2)] $\rhp \geq 0$, $\rhm \geq 0$
		\item[(C3)] $\etpinf + \etminf < \dupinfty - \dmin$
		\item[(C4)] $(1 - \rhm)(\dup'(-\Delta) - \rhp + 1) > 1$.
	\end{enumerate}
	The OR gate fed-back by a strictly causal \etam\ channel has the following 
	output when the input pulse has length $\Delta_0$:
	\begin{itemize}
		\item[(i)] If the input pulse length is $\Delta_0 \geq \dupinfty + 
		\etpinfty$, 
		then the output of the OR gate $o_{or}$ has a unique rising transition 
		at time $0$ and no falling transitions.
		\item[(ii)] If the input pulse length is $\Delta_0 \leq \dupinfty - \dmin - 
		\etpinfty - \etminfty$, then the output pulse of the OR gate $o_{or}$ 
		only contains the input pulse.
		\item[(iii)] If $\dupinfty - \dmin - \etpinfty - \etminfty < \Delta_0 < 
		\dupinfty + \etpinfty$, then the output may resolve to a constant $0$ 
		or $1$, or may be an (infinite) pulse train with a maximum pulse up-time 
		of $\Delta$ and a maximum duty cycle of $\gamma=\Delta/\tau<1$.
	\end{itemize}
\end{theorem}

\begin{proof}
	By combining \cref{lemma:longpulses,lemma:shortpulses,lemma:fixedpoint,thm:mediumpulses},
	the resolution to 1 resp.\ the maximum up-time bound $\Delta$ follow
        immediately. To also confirm the upper bound on the duty cycle, we show
        by contradiction
        that no infinite pulse train can contain a pulse with a down-time
        $\cDelta_n < P-\Delta$: Since $\cP_{n}=\cDelta_n+\Delta_n
        =\ddo(-\cDelta_n)-\etm(-\cDelta_{n})$ by \cref{def:periods-new}, it
        decreases with both $\cDelta_n$ and $\etm(-\cDelta_{n}) \geq \etmmin$.
        Hence, if $\cDelta_n < P - \Delta$ ever occurred, we would get
        $\cP_n > P' = \ddo(-P+\Delta) - \etmmin=P$. Hence, it follows that
        $\Delta_n = \cP_n-\cDelta_n > \Delta$, which 
        contradicts the upper bound $\Delta_n \leq \Delta$ established
        before.
\end{proof}

It only remains to properly choose the high-threshold buffer at the output 
of the OR gate. First, it must map infinite and decreasing pulse trains ($\Delta_n \leq 
\Delta$ for $n\geq 0$) according
to (iii) in \cref{thm:casedistinction} to the zero signal. Second, 
for an increasing pulse train (some $\Delta_n > \Delta$), it needs 
to be ensured that once the 
high-threshold buffer switches to one it never switches back. 
It has been shown in \cite[Lem.~3]{FMNNS18:DATE} that it is
possible to implement a high-threshold buffer with the required
properties with a simple exponential \idm\ channel
(an exp-channel, see \cite{OMFS21:Integration} for details),
which finally establishes:

\begin{theorem}
	There is a circuit that solves unbounded SPF in the extended \etam.
\end{theorem}

\section{The \etacidm}\label{sec:etacidm}

In this section, we will sketch how the extended \etam\ and the 
Composable Involution Delay Model (\cidm) \cite{MOSFN21:GLSVLSI}
can be seamlessly combined. The \cidm\
generalizes the \idm\ by adding a transition-dependent pure-delay shifter 
$\Deltapm$ (with $\Delta^+$ resp.\ $\Delta^-$ affecting $\dup(.)$ resp.\ $\ddo(.)$ only) 
at each input of a gate, and re-ordering/changing the different internal components
of \idm\ channels, as shown in \cref{fig:spfcircuitetacidm} for the SPF circuit in \cref{fig:spfcircuiteta}.
The pure-delay shifter $\Deltapm$ can be used to account for 
threshold voltage variations between gates, which considerably
decreases the circuit characterization effort: a single threshold voltage
like $\vth=\vdd/2$ can be used for all gates of a circuit.
Moreover, the pure-delay shifter makes it easy to implement high-threshold 
buffers with arbitrary threshold voltages.

\begin{figure}[t]
	\centering	
	\ifthenelse{\boolean{conference}}
	{
		\includegraphics[width=1\linewidth]{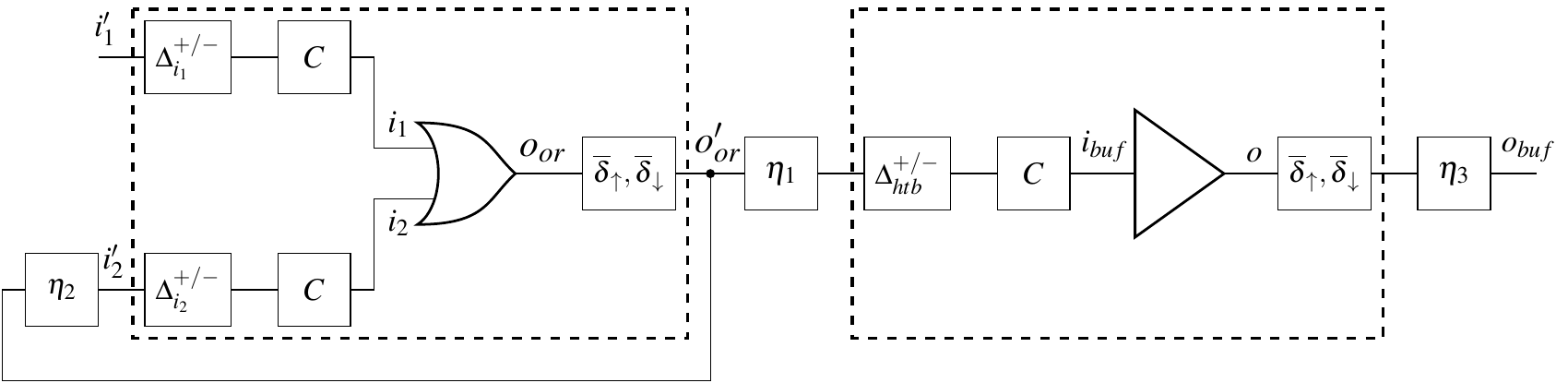}
	}
	{
		\includegraphics[width=1\linewidth]{unbounded_spf_circuit_etacidm2.pdf}
	}
	\caption{The SPF circuit of \cref{fig:spfcircuiteta} implemented by means
of \cidm\ channels. The internal component $\Deltapm$ resp.\ $C$ represent the pure delay 
shifter resp.\ the cancellation unit,
which removes out-of-order transitions. $\bdup,\bddo$ is a standard \idm\ channel without cancellation unit.}
	\label{fig:spfcircuitetacidm}
\end{figure}

\cidm\ channels are characterized by the delay functions
\begin{align}		
	\dup(T) &= \Delta^+ + \bdup(T + \Delta^+) \ \textrm{and}\label{eq:cidmup}\\
	\ddo(T) &= \Delta^- + \bddo(T + \Delta^-),\label{eq:cidmdo}
\end{align}
which do not satisfy the involution property. Still, it turned out that
the result of concatenating the tail (the \idm\ channel $\bdup,\bddo$) 
of the \cidm\ of a predecessor gate with the head (the pure delay shifter 
$\Deltapm$ and the cancellation unit $C$) of the \cidm\ of the successor gate 
\emph{does} form a regular \idm\ channel (albeit with delay functions that are usually
different from $\bdup(.)$ and $\bddo(.)$). Consequently, it is possible to
transfer results obtained for circuits made up of \idm\ channels to 
circuits made up of \cidm\ channels. In particular, this is true for our
extended \etam, as the SPF circuit shown in \cref{fig:spfcircuitetacidm}
is indeed equivalent (modulo some pure delay shifts) to 
\cref{fig:spfcircuiteta}.

In the \etacidm, a \cidm-channel with an appropriately chosen
$\Delta^+$ and $\Delta^-$
can be used to replace the exp-channel-implementation of the high-threshold buffer used in \cite{FMNNS18:DATE}:

\begin{lemma}\label{lemma:htb}
	For any $\Theta > 0$, the high-threshold 
	buffer in \cref{fig:spfcircuitetacidm} can be implemented by a strictly 
	causal \cidm-channel, which ensures that all pulses (at $o_{or}$)
        with up-time $\Delta_n \leq \Theta$ are canceled and the zero signal is
        generated at $i_{buf}$ and hence $o_{buf}$.
\end{lemma}

\ifthenelse{\boolean{conference}}
{
}
{
	\begin{proof}
		The adversaries in the feedback loop ($\eta_2$) and between the OR gate 
		and 
		the high-threshold buffer ($\eta_1$) can add variations in the range 
		[$-{\etminfty}_{,2}, {\etpinf}_{,2}$] resp. [$-{\etminfty}_{,1}, 
		{\etpinf}_{,1}$]. 
		Note that we can neglect the $T$-dependency for the adversaries here 
		and 
		just use the maximum possible delay variation.
		By back-tracking the pulse $\Delta_n$ through ${i_2}$, $i_{2}'$ and the 
		adversary
		$\eta_2$, the maximum length of the corresponding up-pulse at
		$o_{or}'$ can be determined as:
		\begin{equation}
			\Delta_n' = \Delta_n + \Delta_{i_2}^+ + {\etpinf}_{,2} + 
			{\etminf}_{,2} - 
			\Delta_{i_2}^-.
		\end{equation}
		It is achieved if the rising transition is maximally delayed and the 
		falling transition is minimally delayed by the adversary. Consequently,
		we obtain $\Delta_n' \leq \Theta'$ for $\Theta' = \Theta + 
		\Delta_{i_2}^+ + {\etpinf}_{,2} + {\etminf}_{,2} - \Delta_{i_2}^-$. It 
		thus suffices to choose
		$\Delta^+_{htb}$ and $\Delta^-_{htb}$ such that all up-pulses 
		$\Delta_n' \leq \Theta'$ are mapped to constant $0$, which is ensured if
		\begin{equation}\label{eq:htb}
			\Theta' + \Delta^-_{htb} + {\etpinf}_{,1} \leq \Delta^+_{htb} - 
			{\etminfty}_{,1}
		\end{equation}
		holds. Since it has been shown in \cite{MOSFN21:GLSVLSI}
		that $\Delta_{htb}^+$ can be chosen arbitrarily large (whereas
		$\Delta_{htb}^-$ is lower-bounded in order not to violate strict 
		causality
		of the \cidm\ channel), \cref{eq:htb} can always be satisfied. 
		Note that, unlike for the high-threshold buffer implementation
		with the exp-channels used in  \cite{FMNNS18:DATE}, it is easy 
		to determine the necessary parameters via \cref{eq:htb} here. Moreover,
		the upper bound $\gamma$ on the duty cycle established in
		\cref{thm:casedistinction} is not needed anymore.
		
		However, we still need to argue that the additional \idm-channel and
		the adversary $\eta_3$ at the output of the high-threshold buffer
		do not have adverse effects. Fortunately, for the output $o$,
		we only need to distinguish two possible cases: 
		(i) If there is no transition at $o$, then there is obviously also no 
		transition at $o_{buf}$.
		(ii) If there is a single rising transition at the output $o$ at time 
		$T$, 
		then this single transition is propagated to $o_{buf}$ at some time 
		$T+K$. The output $o$ is hence still in accordance with the behavior of 
		a high-threshold buffer.
	\end{proof}
}

We hence obtain our desired result:

\begin{theorem}
	There is a circuit that solves unbounded SPF in the \etacidm.
\end{theorem}

\ifthenelse{\boolean{conference}}
{
}
{
	\begin{proof}	
		If $\Delta_0 \leq \dupinfty - \dmin - \etpinfty - \etminfty$, 
		then the output $o_{or}$ only contains the input pulse $\Delta_0$, 
		according to \cref{thm:casedistinction}. 
		By \cref{lemma:htb}, we know that we can choose $(\Delta^+_{htb}, 
		\Delta^-_{htb})$ such that this pulse $\Delta_0$ is mapped 
		to a constant $0$.
		
		If $\Delta_0 \geq \dupinfty + \etpinfty$, then the output $o_{or}$ has 
		a 
		single rising transition at $t = 0$. 
		Eventually, the output of the high-threshold buffer also has a single 
		transition to $1$.
		
		If $\Delta_0$ is in between the above ranges, two cases can be 
		distinguished: 
		(i) Suppose that the up-pulse lengths never exceed $\Delta$, then
		\cref{lemma:htb} for any $\Theta \geq \Delta$ guarantees that the input 
		is mapped to a constant $0$ at the output of the high-threshold buffer.
		For the other case (ii), the pulse length of the pulse train 
		(eventually) 
		exceeds $\Delta$. 
		By \cref{thm:mediumpulses} it is guaranteed that the feedback loop will 
		eventually lock at constant $1$. 
		If we choose $(\Delta^+_{htb}, \Delta^-_{htb})$ such that all pulses 
		with a 
		width at most $\Theta=\oDelta\geq\Delta$
		are mapped to a constant $0$, we can ensure that 
		the high-threshold buffer has exactly one rising transition after the 
		feedback loop has locked.
	\end{proof}	
}

\section{Simulation results}\label{sec:results}

In this section, we will show\footnote{Due to space constraints, we can only provide
a few of our experimental results.} that the delay predictions of our new \etacidm\ 
cover the actual delays of an inverter chain under substantial PVT variations 
and even under aging.

For our evaluation, we used Cadence Spectre (version 20.1) \cite{Cad20} to 
simulate 
a 7-stage inverter chain in UMC \SI{28}{\nm} technology (G-05-LOGIC/MIXED\_MODE28N-HPC-SPICE). Since their model cards not only offer parameter sets
for different process corners, i.e., different combinations of fast and slow n-MOS resp. p-MOS transistors, but also parameters for simulating \emph{Negative Bias Temperature Instability} (NBTI), \emph{Positive Bias Temperature Instability} (PBTI) and 
\emph{Hot Carrier Injection} (HCI), we could use the RelXpert tool \cite{Spectre:RelXpert:20.1}
to compute
aged versions of the circuit (with transistors suffering from 
threshold voltage shifts and gate oxide damages). 
Since we did not have a gate library for the used \SI{28}{\nm} 
technology, however, we had to design our own inverter cells.
We choose $W_{DES} = \SI{0.35}{\um}$ and $L_{DES} = \SI{0.035}{\um}$, which are 
within the specified range according to the data sheet and between the values 
for comparable \SI{15}{\nm} and \SI{45}{\nm} technologies. 

To compare the delays of the actual circuit and the predictions of our model, 
we had to perform the following steps:
\begin{itemize}
	\item[(1)] Obtain the delay functions $\dupsdef(T)$ and $\ddosdef(T)$ via analog 
	simulations under the default environment (\SI{25}{\degreeCelsius}, 
	$\vdd=\SI{0.9}{\volt}$).
	\item[(2)] Extract the parameters for our delay model, based on $\dupsdef(T)$ 
	and $\ddosdef(T)$.
	\item[(3)] Run analog simulations under PVT variations and aging, and obtain the
        appropriate delay functions $\dupsx{x}(T)$ and $\ddosx{x}(T)$, where 
{\small
        \[
x \in 
\begin{cases}
20a & \mbox{20 years aging},\\
\SI{85}{\degreeCelsius} & \mbox{operating temperature},\\
ss & \mbox{slow nMOS + slow pMOS process corner}, \\
+\SI{10}{\percent}\vdd & \mbox{supply voltage}.
\end{cases}
\]
}
	\item[(4)] Compare the old and the new ``corridor'' of the \etacidm, i.e., its
deterministic delay prediction + the allowed variation range, with the actually measured delay 
functions.
\end{itemize}

To characterize our circuit in steps (1) and (2), i.e., to determine
the parameters $\dmin, \Delta^+$ and $\Delta^-$ and the actual delay function $\dup(.)$ and
$\ddo(.)$ of every gate, 
we used\footnote{The publicly available Involution Tool \cite{OMFS21:Integration} provides characterization scripts, which we extended appropriately to perform
the procedure described below fully automatically.} the approach from 
\cite{MOSFN21:GLSVLSI}. It is illustrated in \cref{fig:characterization_idea}: 
First, running analog simulations for each 
pulse width, a binary search on the input pulse width is used until the output 
trajectory for the falling (green) and rising (orange) transition both touch 
$\vthoutm$. Then, the rising input and the corresponding output trajectory is moved 
in time until the rising and falling output trajectories touch each other at 
$\vthoutm$. Now, $\dmin=t_o-t_i$ is just the time between the touching point 
$t_o$ of the output trajectories and the crossing point $t_i$ of the input 
trajectories. $\Delta^+$ resp.\ $\Delta^-$ is determined as the time for the rising resp.\ 
falling input trajectory to get from $\vthin$ to $\vthinm$. 

\begin{figure}[t]
	\centering
	\ifthenelse{\boolean{conference}}
	{
		\includegraphics[width=0.85\linewidth]{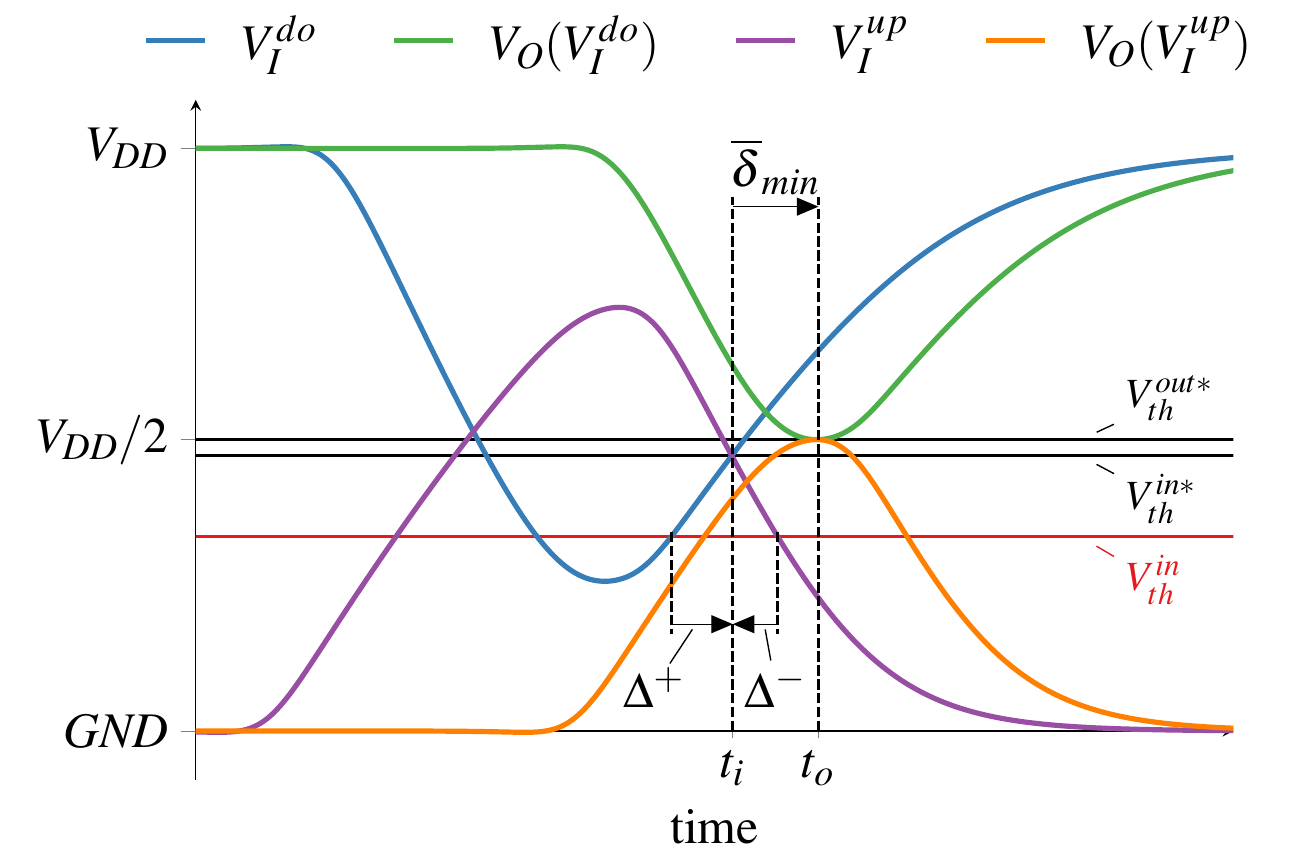}
	}
	{
		\includegraphics[width=1.0\linewidth]{buffer_vthin.pdf}
	}
	\caption{Characterization idea for a buffer (taken from 
		\cite{MOSFN21:GLSVLSI}).}	
	\label{fig:characterization_idea} 
\end{figure}

In order to determine $\dup(.)$ and $\ddo(.)$, we first measured the actual
delays, by sending longer pulses through the circuit and recording
the pair $(T, \delta(T))$ for each input/output trajectory.
Then, the parameters of a SumExp-\idm\ channel
\cite{OMFS21:Integration} were fitted to these empirical delay functions
using a least squares approach.
For the adversary, we choose symmetric parameters, i.e., $\etpmin = \etmmin, 
\etpinf = \etminf$ and $\rho^+ = \rho^-$.

For step (3), we had to repeat the characterization of the circuit for the
PVT variations and aging considered.

In step (4), all simulated delay functions were compared to the corridor of
the delay predictions of the \etacidm: \cref{fig:results_inv4-inv5_delta_up_sumexp_default} 
shows the results for the rising delay function between the fourth and fifth 
inverter of the circuit.
The delay function $\dupsdef(T)$ (blue) is the baseline, obtained via analog 
simulations. The deterministic delay function $\dup(T)$ of the \etacidm\ 
(red) has been least squares fitted to $\dupsdef(T)$. 
The old (narrow) corridor is shown in green, whereas the significantly larger 
new corridor is lilac.

It is immediately apparent that it was already difficult to fit a channel to the 
baseline (without PVT variations and aging) when using the old bounds 
(compare \cite[Fig.~9]{FMNNS18:DATE}), whereas this is 
absolutely no problem for our bounds. 

\begin{figure}[t]
	\centering
	\ifthenelse{\boolean{conference}}
	{
		\includegraphics[width=1\linewidth]{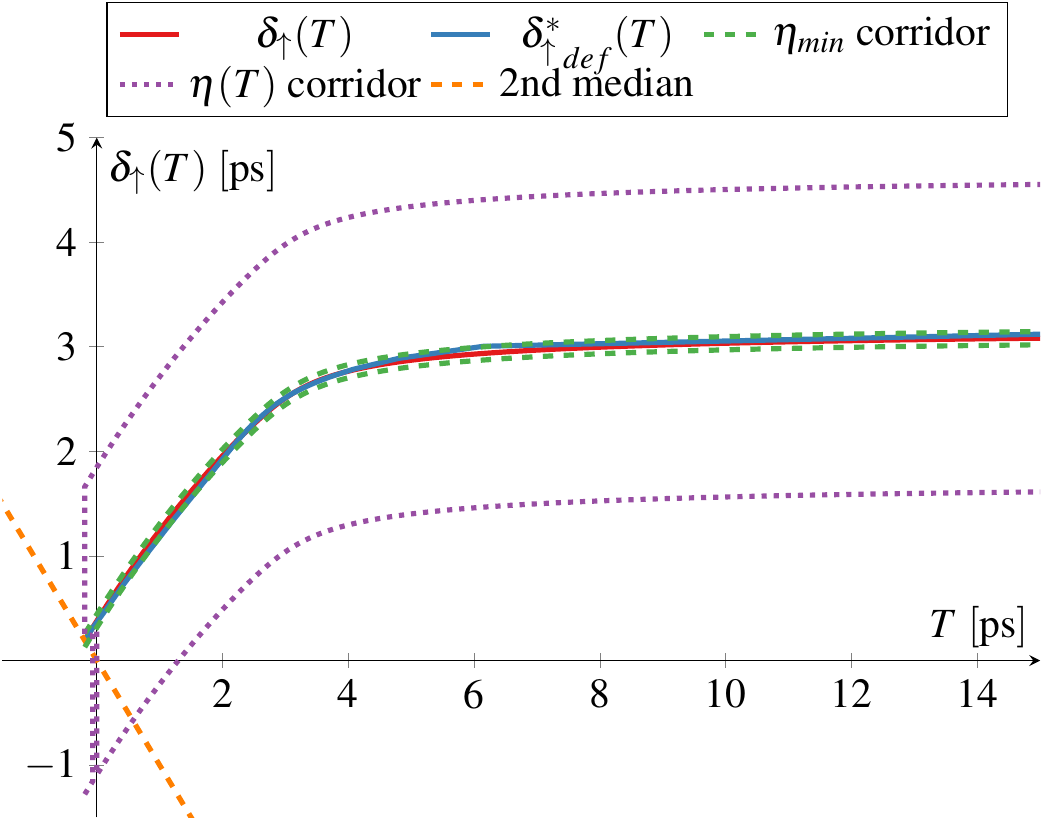}
	}
	{
		\includegraphics[width=1\linewidth]{results_inv4-inv5_delta_up_sumexp_default_dmin.pdf}
	}
	\caption{Actual (measured) delay function $\dupsdef(T)$ between the fourth 
		and fifth inverter of the simulated inverter chain.
		It is compared with the calculated delay function $\dup(T)$.}
	\label{fig:results_inv4-inv5_delta_up_sumexp_default} 
\end{figure}

\cref{fig:results_inv4-inv5_delta_up_sumexp_variations} shows the 
PVT variations and aging results for the rising delay function between the 
fourth and fifth inverter.
Note that the results for the other gates in the circuit are 
comparable.
The delay function with 10~\% increased supply voltage is obviously faster than 
the one under the default environment.
For an increased temperature of \SI{85}{\celsius}, the resulting delay 
function is only slightly slower.
The same also holds for the circuit that is aged by 20 years.
\cref{fig:results_inv4-inv5_delta_up_sumexp_variations} also
shows the slow down for 
the process corner \emph{ss}, i.e. a slow n-MOS and a slow p-MOS 
transistor.
It is apparent that all these PVT variations and aging are covered
by our new model.

\begin{figure}[t]
	\centering
	\ifthenelse{\boolean{conference}}
	{
		\includegraphics[width=1\linewidth]{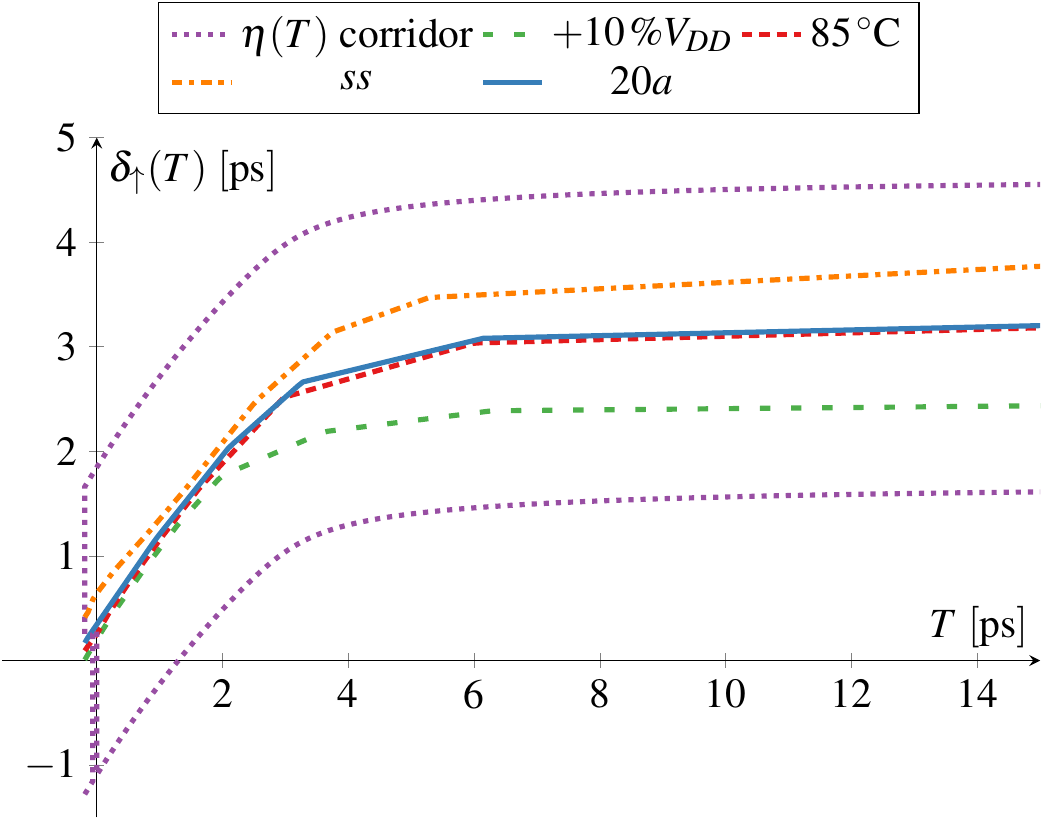}
	}
	{
		\includegraphics[width=1\linewidth]{results_inv4-inv5_delta_up_sumexp_dmin.pdf}
	}
	\caption{Coverage of various PVT variations and aging under the new bounds 
	of the \etam.}
	\label{fig:results_inv4-inv5_delta_up_sumexp_variations} 
\end{figure}

To complement the above results, we also evaluated the average coverage 
of the \etacidm\ for the entire circuit.
More specifically, we considered both $\dup(T)$ and $\ddo(T)$ for all seven 
inverters and numerically integrated over the difference between 
the measured delays and the respective corridor, for every recorded value of $T$.
If the delay function is inside the corridor, the contribution to the integral 
is $0$, otherwise it is the distance to the closest corridor border.
By normalizing the resulting value to the integration interval, we obtain 
a value that can be viewed as the average coverage violation, i.e., the
average time deviation from the corridor.

\cref{fig:results:results_total_values_vdd_persample} shows the results for 
supply voltage variations. 
It can be seen that the average deviation is significantly larger for the old 
bounds, whereas there is almost no deviation for the new bounds.  
A significant deviation is observed only for \SI{-20}{\percent} $\vdd$, 
which is unlikely to happen in modern VLSI circuits.
\begin{figure}[t]
	\centering
	\ifthenelse{\boolean{conference}}
	{
		\includegraphics[width=0.85\linewidth]{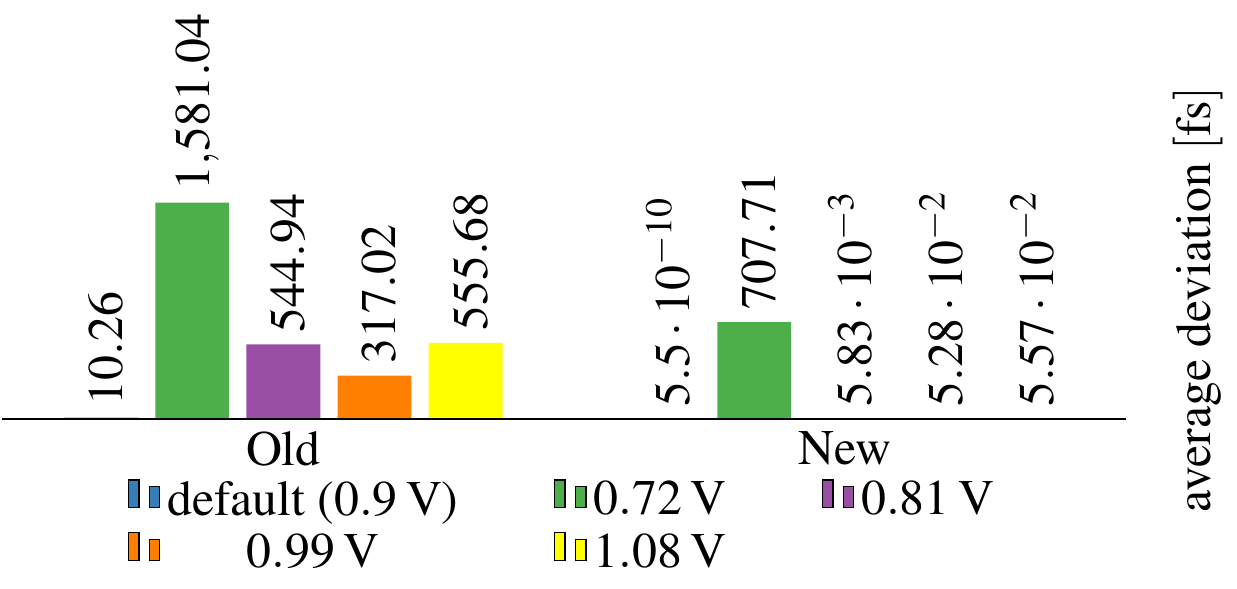}
	}
	{
		\includegraphics[width=0.9\linewidth]{results_total_values_vdd_integral.pdf}
	}
	\caption{Average deviation for different supply voltages ($\vdd$).}
	\label{fig:results:results_total_values_vdd_persample} 
\end{figure}

Like in the \etam-paper \cite{FMNNS18:DATE}, we also considered transistor
width variations. Again, \cref{fig:results:results_total_values_pvt2_persample} reveals 
that our new bounds perfectly cover even substantial variations, unlike the
old bounds, which result in significant deviations (albeit not as large as for 
$\vdd$ variations).

\begin{figure}[t]
	\centering
	\ifthenelse{\boolean{conference}}
	{
		\includegraphics[width=0.85\linewidth]{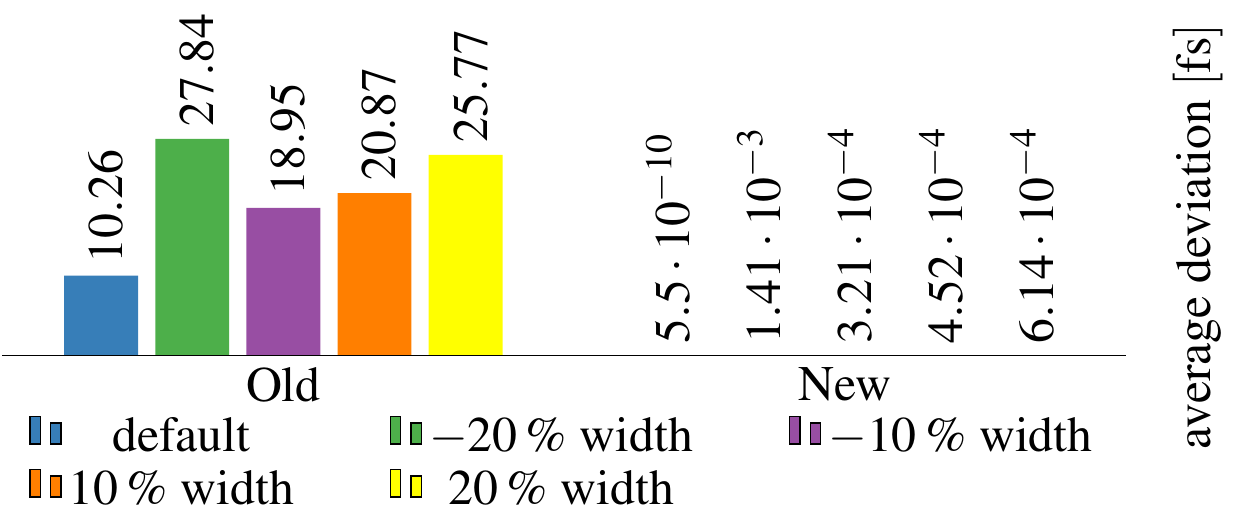}
	}
	{
		\includegraphics[width=0.9\linewidth]{results_total_values_pvt2_integral.pdf}
	}
	\caption{Average deviation for different values of $W_{DES}$.}
	\label{fig:results:results_total_values_pvt2_persample} 
\end{figure}

Overall, our results confirm that the \etacidm\ indeed surpasses the coverage of 
the original \etam\ of \cite{FMNNS18:DATE} substantially. 

\section{Conclusion}\label{sec:conclusion}

We provided a new unbounded single-history delay model, the \etacidm, which supports
a large range of adversarial delay variations.
It substantially improves on the existing \etam, by making the delay variation 
range dependent on every transitions' particular input-to-previous-output time. 
We proved analytically that this extension does not invalidate the implementability of
SPF, hence preserves faithfulness. By means of exhaustive simulations,
we showed that our \etacidm, unlike the original \etam, can cover a substantial range 
of PVT variations and aging.

\bibliographystyle{IEEEtran}
\bibliography{mybib}

\end{document}